\RequirePackage[l2tabu, orthodox]{nag}
\documentclass{llncs}
\pdfoutput=1

\usepackage{etex,booktabs,subcaption,amssymb,tikz,multirow,algorithm,algpseudocode,verbatim,setspace,stfloats,xspace, amsthm, expl3}
\usepackage{cite}
\usepackage[textsize=small]{todonotes}
\usepackage{microtype}
\usetikzlibrary{shapes}
\usetikzlibrary{calc}

\newtheorem{observation}[theorem]{Observation}

\newcommand{\norm}[1]{\left\Vert#1\right\Vert}

\newcommand{\Real}{\mathbb R}
\newcommand{\eps}{\varepsilon}

\newcommand{\Frechet}{Fr\'echet\xspace}

\newcommand{\opt}{\mathrm{opt}}
\newcommand{\apx}{\mathrm{apx}}
\newcommand{\crit}{\mathrm{crit}}
\newcommand{\df}{\delta_\mathcal{F}}
\newcommand{\dF}{\delta_\mathrm{F}}
\usetikzlibrary{shapes}
\usetikzlibrary{calc}
\newcommand{\tikzsetnextfilename}[1] { }

\makeatletter
\renewcommand{\todo}[2][]{\tikzexternaldisable\@todo[#1]{#2}\tikzexternalenable}
\newcommand{\getxy}[3]{%
  \tikz@scan@one@point\pgfutil@firstofone#1\relax
  \edef#2{\the\pgf@x}%
  \edef#3{\the\pgf@y}%
}
\makeatother

\makeatletter
\pgfmathdeclarefunction{atan3}{2}{%
	\begingroup
		\pgfmathsetmacro \pgfmathx {abs(#1)}%
		\pgfmathsetmacro \pgfmathy {abs(#2)}%
		\ifdim \pgfmathx pt>\pgfmathy pt\relax%
			\pgfmathparse{atan2(#1,#2)}%
		\else
			\pgfmathparse{90 - atan2(#2,#1)}%
		\fi%
 		\pgfmath@smuggleone\pgfmathresult%
	\endgroup
}
\makeatother

\newcommand{\tikzdefines}[0] {
    \tikzstyle{cblack}=[circle, draw, thick, solid, fill=black, scale=.15]
    \tikzstyle{cblue}=[circle, draw, solid, black, thin, fill=cyan!20, scale=.4]
    \tikzstyle{cred}=[red, circle, draw, solid, fill=red!20, scale=.4]
    \tikzstyle{cgreen}=[circle, draw, solid, fill=green!60, scale=.4]
}
\pgfdeclarelayer{background2}
\pgfdeclarelayer{background}
\pgfsetlayers{background2,background,main}

\newcommand{\pathborder}[2]{
    [
        create hullcoords/.code={
            \global\edef\namelist{#1}
            \foreach [count=\counter] \nodename in \namelist {
                \global\edef\numberofnodes{\counter}
                \coordinate (hullcoord\counter) at (\nodename);
            }
            \pgfmathtruncatemacro \numberofnodes {\numberofnodes - 1}
            \foreach [count=\counter] \nodenum in {\numberofnodes,...,1} {
                \pgfmathtruncatemacro \next {\counter + \numberofnodes + 1}
                \coordinate (hullcoord\next) at (hullcoord\nodenum);
            }
            \pgfmathtruncatemacro \numberofnodes {\numberofnodes * 2 + 1}
            \coordinate (hullcoord0) at (hullcoord2);
            \pgfmathtruncatemacro \lastnumber {\numberofnodes+1}
            \coordinate (hullcoord\lastnumber) at (hullcoord2);
        },
        create hullcoords
    ]
    ($(hullcoord1)!#2!90:(hullcoord0)$) node[circle] {}
    \foreach [
        evaluate=\currentnode as \prevnode using \currentnode-1,
        evaluate=\currentnode as \nextnode using \currentnode+1
    ] \currentnode in {2,...,\numberofnodes} {
        let
            \p1 = ($(hullcoord\currentnode) - (hullcoord\prevnode)$),
            \n1 = {atan3(\x1,\y1) + 90pt},
            \p2 = ($(hullcoord\nextnode) - (hullcoord\currentnode)$),
            \n2 = {atan3(\x2,\y2) + 90pt},
            \n3 = {Mod(\n2-\n1,360) - 360},
            \n4 = {cos(.5*(\n3+360))},
            \n5 = {\n1+.5*\n3}
        in
        {
            \ifdim \n3 < -180.05pt {
                -- ($(hullcoord\currentnode) + (\n5:-1/\n4*#2)$)
            } \else {
                -- ($(hullcoord\currentnode)!#2!-90:(hullcoord\prevnode)$)
                arc [start angle=\n1, delta angle=\n3, radius=#2]
            } \fi
        }
    }
}

\makeatletter
\newcommand*\latinabbrev[1]{%
    \@ifnextchar{.}%
        {#1}%
        {#1.\@\xspace}%
}
\makeatother

\def\etal{\latinabbrev{et al}}

\renewenvironment{thebibliography}[1]{%
\begin{oldthebibliography}{#1}%
\setlength{\itemsep}{0ex}%
}%
{%
\end{oldthebibliography}%
}

\title{Finding a Curve in a Point Set}
\institute{
Dept. of Computer \& Info. Sci. \& Eng.\\
University of Florida, Gainesville, FL, USA\\
{\tt \{accisano, ungor\}@cise.ufl.edu}
}
\author{Paul Accisano \and Alper {\"{U}ng\"{o}r}}

\begin{document}

\maketitle

\begin{abstract}
    Let $P$ be a polygonal curve in $\mathbb{R}^D$ of length $n$, and $S$ be a point set of size $k$.  The Curve/Point Set Matching problem consists of finding a polygonal curve $Q$ on $S$ such that its \Frechet distance from $P$ is less than a given $\eps$.  In this paper, we consider this problem with the added freedom to transform the input curve $P$ by translating it, rotating it, or applying an arbitrary affine transform.  We present exact and approximation algorithms for several variations of this problem.
\end{abstract}

\section{Introduction}

Matching a curve and a set of points is a typical geometric problem that arises in many science and engineering fields, such as computer aided design, computer graphics and vision, and protein structure prediction \cite{Brakatsoulas05,Jiang08}. In these applications, data is typically gathered as a point set through a scanner, and the goal is often to find certain objects, described as polygonal curves, in the scene.  In order to perform a matching, one needs a similarity metric between geometric constructs.  In this paper, we study the problem of curve and point set matching, using the \Frechet distance as the similarity metric.

\begin{table}[b]
    \centering
    \begin{tabular}{l@{\hspace{1em}}l@{\hspace{1em}}|@{\hspace{3em}}cl@{\hspace{3em}}c@{\hspace{-1em}}c}
        \toprule
                &               &   \multicolumn{2}{l}{Discrete} & Continuous\\
        \midrule
        Subset  & Unique        & NP-C & \cite{Wylie12}     & NP-C & \cite{Accisano14} \\
                & Non-Unique    & P    & \cite{Wylie12}     & P & \cite{Maheshwari11} \\
        All-Points & Unique        & NP-C & \cite{Wylie12}     & NP-C & \cite{Accisano12}\\
                & Non-Unique    & P    & \cite{Wylie12}     & NP-C & \cite{Accisano12} \\
        \bottomrule
    \end{tabular}%
    \bigskip
    \caption{Eight versions of the CPSM problem and their complexity classes. \vspace{-10pt}}
    \label{tab:results}%
\end{table}%

Given a point set and a polygonal curve, the goal is to connect the points into a new polygonal curve that is similar to the given curve.  An important factor is whether or not the input curve is allowed to be translated or transformed.  The problem in which the curve is fixed in place has been well-studied in the literature \cite{Accisano12,Maheshwari11,Wylie12}, and we refer to it in this paper as the \textbf{\emph{Curve/Point Set Matching (CPSM)}} problem.  Formally, given a polygonal curve $P$ of length $n$, a point set $S$ of size $k$, and a real number $\eps > 0$, determine whether there exists a polygonal curve $Q$ on a subset of the points of $S$ such that $\df(P, Q) \le \eps$.

Eight versions of the original CPSM problem can be classified based on whether the use of all points is enforced, whether points are allowed to be visited more than once, and whether the \Frechet distance metric used is discrete or continuous.  Table \ref{tab:results} summarizes the versions and their known complexity classes.

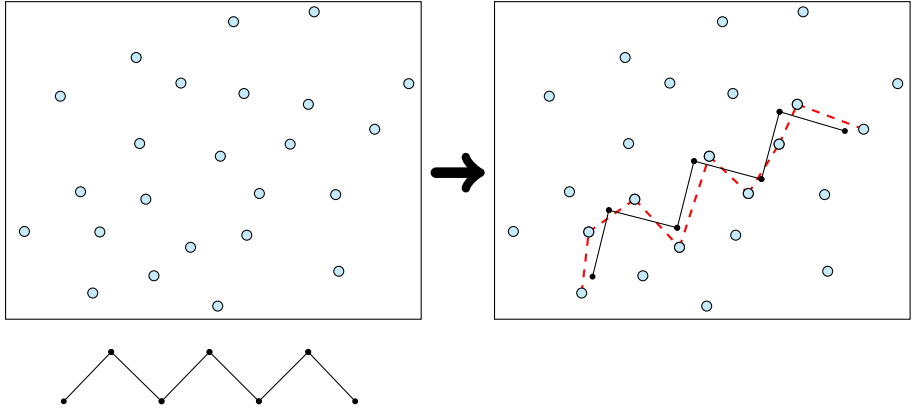
\begin{figure}[t]
    \centering
    \makebox[\textwidth][c]{
    \tikzsetnextfilename{fig_example}
\begin{tikzpicture}[scale=.65]
    \tikzdefines
    
    \newcommand{\vertP}{(6.15,0.32), (5.19,1.33), (4.18,0.32), (3.17,1.33), (2.19,0.32), (1.16,1.33), (0.19,0.32)}
    \newcommand{\vertQ}{(-0.17,0.14), (0.58,1.15), (1.73,1.26), (2.03,-0.05), (3.49,1.26), (3.8,0.2), (4.85,0.76), (5.58,1.28), (6.5,0.16), (1.09,-0.18), (0.65,2.06), (2.19,2.31), (3.15,-0.41), (1.91,-1.37), (7.57,0.62), (6.63,2.86), (4.55,2.13), (5.14,-0.6), (-0.75,1.93), (1.27,3.96), (3.01,3.87), (5.1,3.51), (4.41,-1.99), (3.54,2.96)}

    \def\qn{24};

    \begin{scope}[shift={(-10,0)}]
        \begin{scope}[rotate=30]
            \foreach [count=\x] \pt in \vertQ {
                \coordinate [at=\pt, name=q\x];
            }
            \foreach \x in {1, ..., \qn} {
                \draw (q\x) node[cblue] {};
            }
        \end{scope}
        
       \begin{scope}[shift={(-1, -2.5)}]
            \foreach [count=\x] \pt in \vertP {
                \coordinate [at=\pt, name=p\x];
            }
            \foreach \x  [evaluate=\x as \y using (\x-1)] in {2, ..., 7} {
                \draw (p\y) node[cblack] {} --(p\x) node[cblack] {};
            }
        \end{scope}
        \draw[line width=4pt,arrows=->,cap=round] (6.8,2.5) -- (7.8,2.5); 
        \draw (-2,6) rectangle (6.5, -.5);
    \end{scope}

    \begin{scope}[rotate=30]
        \foreach [count=\x] \pt in \vertQ {
            \coordinate [at=\pt, name=q\x];
        }
        \foreach \x  [evaluate=\x as \y using (\x-1)] in {2, ..., 9} {
            \draw[thick, dashed, red] (q\y) node[cblue] {} --(q\x) node[cblue] {};
        }
        \foreach \x in {10, ..., \qn} {
            \draw (q\x) node[cblue] {};
        }

        \foreach [count=\x] \pt in \vertP {
            \coordinate [at=\pt, name=p\x];
        }
        \foreach \x  [evaluate=\x as \y using (\x-1)] in {2, ..., 7} {
            \draw (p\y) node[cblack] {} --(p\x) node[cblack] {};
        }
    \end{scope}
    \draw (-2,6) rectangle (6.5, -.5);
\end{tikzpicture} 
    }
    \caption{Example TCPSM instance and solution.}
    \label{fig:example}
\end{figure}

\Frechet distance is a powerful and useful metric, but it is very sensitive to positional and rotational differences.  If the goal is to locate a curve in a point set that is similar to a given curve, many applications would want the shape of the curve to be the only relevant factor, not its position or orientation.  One way to achieve this is to allow the curve to be transformed by an affine transformation from a user-specified set.  Indeed, in the literature, the problem of matching two curves under a specified set of affine transformations has been well-studied \cite{Wenk02,Jiang08}.  In this paper, we introduce the \textbf{\emph{Transformed Curve/Point Set Matching (TCPSM)}} problem, in which the goal is to find the transformation of $P$ (limited to a user-specified set) that puts it closest to some curve whose vertices are in $S$.  Figure \ref{fig:example} shows an example instance.  We note that, in this paper, we focus on the Non-unique versions only, and we use the convention that all the TCPSM problems discussed henceforth refer to the Non-unique version.

\textbf{Our results.}  We present an algorithm that makes use of the results in \cite{Maheshwari11} and \cite{Wenk02} to solve the Continuous Non-Unique TCPSM.  The algorithm runs in $O((nk)^{2d+1}k)$ time for the decision version of the problem, where $d$ is the number of degrees of freedom in the affine transform matrix.  Similarly, we show how the Discrete versions are solvable in a similar fashion, resulting in $O((nk)^{d+1})$ time algorithms for both the Subset and All-Points versions.  Finally, we give a new algorithm designed from scratch to solve the special case of translations in $\Real^2$ for the Discrete versions.  Our algorithm runs in $O(n^2k^2\log(nk))$ time, which is an improvement over the general case.  Interestingly, this is faster than the best known algorithm for finding the optimal translation when both curves are given.  The optimization version of all the problems discussed can be solved with an additional $O(\log(nk))$ factor via parametric search.

\section{Preliminaries}
Below, we present the notation that will be used throughout the paper, some of which is similar to the notation used by earlier work \cite{Accisano12, Alt03,Maheshwari11}.  More will be introduced later as needed.  Given two curves $P, Q : [0, 1] \rightarrow \mathbb{R}^d$, the \emph{\Frechet distance} between $P$ and $Q$ is defined as 
$$
\df(P, Q) = \inf_{\sigma, \tau} \max_{t\in[0,1]} \norm{P(\sigma(t)), Q(\tau(t))}$$
where $\sigma, \tau : [0, 1] \rightarrow [0, 1]$ range over all continuous non-decreasing surjective functions \cite{Ewing85}.  Deciding whether two curves have \Frechet distance less than a given $\eps$ can be done in $O(nm)$ time, and finding the actual \Frechet distance can be determined in $O(nm \log(nm))$ time by applying parametric search \cite{Alt95}.  The Continuous Subset CPSM, which uses the continuous \Frechet distance, is solvable in $O(nk^2)$ time \cite{Maheshwari11}.

Discrete \Frechet distance is a variation of the standard \Frechet distance that only takes into account distance at the curve vertices \cite{Eiter94}.  For two curves $P$ and $Q$ of lengths $n$ and $m$ respectively, a \emph{paired walk} or \emph{coupling sequence} is a pair of integer sequences $(a_1, b_1), \dots, (a_k, b_k)$, $k \ge \max(n, m)$, with the properties that $(a_1, b_1) = (1, 1), (a_k, a_k) = (n, m)$, and for all $i$, $(a_{i+1}, b_{i+1}) \in \{(a_{i}+1, b_i), (a_i, b_{i}+1), (a_{i}+1, b_{i}+1)\}$.  Let $W$ be the set of all paired walks for $P$ and $Q$.  Then the \emph{discrete \Frechet distance} can be defined as:
$$
\dF = \min_{(a, b) \in W} \max_i \norm{P_{a_i}, Q_{b_i}}
$$
Note that $\df(P, Q) \le \dF(P, Q)$.  The discrete \Frechet distance can be computed directly, without the need for parametric search, in $O(nm)$ time via a dynamic programming algorithm \cite{Eiter94}.  Recently, Agarwal \etal \cite{Agarwal12} presented an algorithm that finds the discrete \Frechet distance in $O(nm \log \log(nm) / \log(nm))$ time, breaking the quadratic barrier for this problem.  The best known algorithm for the discrete CPSM is still $O(nk)$, for both decision and optimization versions \cite{Wylie12}.

A polygonal curve $P$ is defined by a set of vertices $P_0, \dots, P_n$.  We use $\overline{P_i}$ to refer to the segment of $P$ between $P_{i-1}$ and $P_i$, and we use $V(P)$ to denote the set of vertices of $P$.  Let $P$ and $Q$ be curves and $\eps = \df(P, Q)$. A pair of points $(p, q) \in (P \times Q)$, residing in segments $\overline{P_i}$ and $\overline{Q_j}$ respectively, is said to be \emph{feasible} if the subcurves $(P_0, \dots, P_{i-1}, p)$ and $(Q_0, \dots, Q_{j-1}, q)$ have \Frechet distance at most $\eps$.  If $(p, q) \in (\overline{P_i} \times Q)$ is feasible, we say that $q$ is \emph{visited at} segment $i$ of $P$.  For every point $q \in Q$, there must be at least one feasible pair $(p, q)$, and thus every point in $Q$ is visited in at least one segment of $P$.

\section{Previous Work}
\subsection{Maheshwari's Algorithm}
In \cite{Maheshwari11}, Maheshwari \etal gave a dynamic programming algorithm to solve the Continuous Subset CPSM in polynomial time.  The algorithm relies on a previous result by Alt \etal \cite{Alt03} to compute the reachability information for every pair of points in $S$ with respect to $P$.  Starting from the first vertex of $P$, the algorithm works by computing which points are reachable at which segments, propagating this reachability information through the curve.  Once the final vertex of $P$ is reached, if any points in $S$ are still reachable, the algorithm returns true.  It runs in $O(nk^2)$ time.  We need not modify Maheshwari's algorithm for our purposes; we simply call it as a subroutine.

\subsection{Wenk's Algorithm}
Wenk \cite{Wenk02} looked at a problem similar to the TCPSM.  In that problem, instead of matching a curve to a point set, they match a curve to another curve, allowing one to be transformed in order to find the optimal \Frechet distance.  As we will make use of this algorithm, we give a full overview in this section.

Let $\mathcal{T}$ be a rationally parameterized\footnote{See \cite{Wenk02} for a full definition of ``rationally parameterized.''  The set of rationally parameterized affine transforms includes the set of all commonly used transformation types, such as translations, rotations, similarities, and arbitrary affine transforms.} subset of the set of affine transformations operating in $\Real^D$.  Each transformation in $\mathcal{T}$ can be represented by a $D \times D$ linear transformation matrix $A$ and a translation vector $t \in \Real^D$.  Thus, the number of degrees of freedom $d$ is at most $D^2 + D$, and the parameter space of  $\mathcal{T}$ can be represented by $\Real^d$.  In \cite{Wenk02}, Wenk considers the problem of finding the $x \in \Real^d$ that minimizes $\df(A_x P + t_x, Q)$, where $P$ and $Q$ are polygonal curves.  In keeping with many other \Frechet results, the strategy employed is to first solve the decision version of the problem, which asks if there exists an $x \in \Real^d$ for which $\df(A_x P + t_x, Q) \le \eps$, and then use parametric search to find the minimum.  For brevity, we let $\tau_x(P) = A_x P + t_x$.

\begin{lemma}\emph{\hspace*{-.07in}\cite{Wenk02}} \label{lem:continuous}
    For any two polygonal curves $P$ and $Q$, $\df(\tau_x(P), Q)$ is continuous as a function of $x \in \Real^d$.
\end{lemma}

\begin{corollary}\emph{\hspace*{-.07in}\cite{Wenk02}}
    For a given $\eps > 0$, if there exists some $x \in \Real^d$ for which $\df(\tau_x(P), Q) < \eps$, then there exists some $x \in \Real^d$ for which $\df(\tau_x(P), Q) = \eps$
\end{corollary}

\begin{lemma}\emph{\hspace*{-.07in}\cite{Wenk02}} \label{lem:config}
    If $\df(P, Q) = \eps$, then either:
    \begin{itemize}
    \item there exists a vertex $x$ from one curve and a segment $\overline{Z}$ from the other curve for which $\min_{z \in \overline{Z}} \norm{x, z} = \eps$,
    \item or there exist two vertices $x$ and $y$ from one curve and a segment $\overline{Z}$ from the other curve for which $\norm{x, z} = \norm{y, z} = \eps$ for some $z \in \overline{Z}$.
    \end{itemize}
\end{lemma}

The above lemma shows that, in order to find a transformation of $P$ which puts its \Frechet distance from $Q$ at exactly $\eps$, we need only consider those transformations which cause one of these two conditions to arise.  Both involve a segment and at most two vertices.  This leads to the introduction of configurations.

\begin{definition}\emph{\hspace*{-.07in}\cite{Wenk02}}
A \emph{configuration} is a triple $(x, y, \overline{Z})$, where $x$ and $y$ are vertices from one curve (possibly the same vertex) and $\overline{Z}$ is a segment from the other.
\end{definition}

\begin{definition}\emph{\hspace*{-.07in}\cite{Wenk02}}
For a given configuration $c = (P_i, P_j, \overline{Q_k})$, the set of \emph{critical transformations} $T^\eps_\crit(c)$ is defined as:
$$
T^\eps_\crit(c) =
\left\{
  \begin{array}{lr}
    \{x \in \Real^d \mid \min_{z \in \overline{Q_k}} \norm{\tau_x(P_i), z} = \eps \}& : i = j\\
    \{x \in \Real^d \mid \exists z \in \overline{Q_k} \, \norm{\tau_x(P_i), z} = \norm{\tau_x(P_j), z} = \eps \}& : i \ne j
  \end{array}
\right.
$$

Symmetrically, if $c = (Q_i, Q_j, \overline{P_k})$,
$$
T^\eps_\crit(c) =
\left\{
  \begin{array}{lr}
    \{x \in \Real^d \mid \min_{z \in \tau_x(\overline{P_k})} \norm{Q_i, z} = \eps \}& : i = j\\
    \{x \in \Real^d \mid \exists z \in \tau_x(\overline{P_k}) \, \norm{Q_i, z} = \norm{Q_j, z} = \eps \}& : i \ne j
  \end{array}
\right.
$$
A transformation $\tau_x$ is said to be \emph{critical} if it is critical for some configuration $c$.  The arrangement in $\Real^d$ of all critical transformations is denoted by $A^\eps_\crit$.
\end{definition}


\begin{lemma}\emph{\hspace*{-.07in}\cite{Wenk02}} \label{lem:face}
If there exists a transformation $\tau_x$ such that $\df(\tau_x(P), Q) < \eps$ then there exists some face $F \in A^\eps_\crit$ such that $\df(\tau_y(P), Q) \le \eps$ for all $y \in F$.
\end{lemma}

After showing these lemmas, the strategy employed by Wenk \cite{Wenk02} is to obtain a sample point from each face in $A^\eps_\crit$ and check the \Frechet distance for each corresponding transformation.  To obtain such a sample, a result of Basu, Pollack, and Roy \cite{Basu98} is employed.  This result shows how to obtain a sample point from each face of an arrangement of semi-algebraic sets in $d$-dimensional space in $O(M^d)$ time and space, where $M$ is the number of sets.  The resulting sample set is termed a \emph{semi-algebraic sample}.  Of course, this does require one last lemma.

\begin{lemma}\emph{\hspace*{-.07in}\cite{Wenk02}} \label{lem:semialg}
For any configuration $c$, $T^\eps_\crit(c)$ is semi-algebraic.
\end{lemma}

Since there are a total of $n^2m + nm^2$ possible configurations, and since deciding if the \Frechet distance of two curves is less than $\eps$ takes $O(nm)$ time \cite{Alt95}, the total complexity of Wenk's algorithm is $O((nm)^{d+1}(n+m)^d)$.

\section{Exact Algorithms for the TCPSM} \label{sec:tcpsm}
\subsection{Continuous Subset versions}
Wenk's algorithm relies on the concept of a configuration: a triple of two vertices from one curve and a segment from the other.  The important observation, however, is that each configuration is independent of the others, and the order of the vertices is not relevant to the construction of the arrangement.  Therefore, to solve the TCPSM in which one of the curves is unknown, computing a superset of the configurations of a valid curve is sufficient.  We can do this by considering all points and potential edges in $S$.  The configurations corresponding to a valid curve $Q$, if there is one, will be among them.  The addition of extra configurations to the arrangement does not affect the correctness of the algorithm.

Since there are $k^2$ potential edges in the point set, the number of configurations $M$ is $O(nk^2 + n^2k^2) = O(n^2k^2)$.  The semi-algebraic sample can be computed in $O(M^d)= O((nk)^{2d})$.  Finally, each possible transformation must be checked with Maheshwari's algorithm, which takes $O(nk^2)$ time  \cite{Maheshwari11}.  Thus, the total running time for the decision version is $O((nk)^{2d+1}k)$.  The optimization version can be solved using parametric search, which adds an additional log factor, leading to a running time of $O((nk)^{2d+1}k\log(nk))$.

\subsection{Discrete Subset and All-points versions}
The Discrete versions of the TCPSM can be solved in a similar fashion, but the lemmas above need to be reexamined to ensure they hold for the discrete \Frechet distance.  Lemma \ref{lem:continuous} trivially holds from the fact that the discrete \Frechet distance is defined by a set of minimums and maximums of Euclidean distance functions. 

\begin{lemma} \label{lem:dcontinuous}
For any two polygonal curves $P$ and $Q$, $\dF(\tau_x(P), Q)$ is continuous as a function of $x \in \Real^d$.
\end{lemma}

Lemma \ref{lem:config} is not so easy; it only holds true for the continuous \Frechet distance.  Its discrete counterpart is as follows.
\begin{lemma} \label{lem:dconfig}
If $\dF(P, Q) = \eps$, then there exist two vertices $P_i$ and $Q_j$ for which $\norm{P_i, Q_j} = \eps$.
\end{lemma}
\begin{proof}
The proof is much simpler that the corresponding continuous Lemma, owing to the fact that the number of unique paired walks is finite.  Since the expression inside the minimum and maximum of the discrete \Frechet distances' definition is the distance between two vertices, $\dF$ must be one of these distances. 
\end{proof}

Thus, configurations for the discrete problem are redefined as a pair $(P_i, Q_j)$, and critical transformations of a configuration $c = (P_i, Q_j)$ are redefined as:
$$
T^\eps_\crit(c) = \{x \in \Real^d \mid \norm{\tau_x(P_i), Q_j} = \eps \}
$$

The discrete \Frechet equivalent of Lemma \ref{lem:semialg} trivially holds for this new definition, as $T^\eps_\crit(c)$ is defined by a simple polynomial expression.  The discrete \Frechet equivalent of Lemma \ref{lem:face} also holds, but although the proof is very similar to proof for the continuous version given in \cite{Wenk02}, we feel it is different enough to warrant at least a sketch of the proof, highlighting the differences between the two.

\begin{lemma}
If there exists a transformation $\tau_x$ such that $\dF(\tau_x(P), Q) < \eps$ then there exists some face $F \in A^\eps_\crit$ such that $\dF(\tau_y(P), Q) \le \eps$ for all $y \in F$.
\end{lemma}

\begin{proof}
    By the continuity of discrete \Frechet distance and its corollary, the existence of a transformation $\tau_x$ for which $\dF(\tau_x(P), Q) < \eps$ implies the existence of a transformation $\tau_{x^=}$ for which $\dF(\tau_{x^=}(P), Q) = \eps$.  By Lemma \ref{lem:dconfig}, there is some configuration $c$ for which $x^= \in T^\eps_\crit(c)$.  Let $F$ be the connected component of  $T^\eps_\crit(c)$ that contains $x^=$.  If  $\dF(\tau_y(P), Q) \le \eps$ for all $y \in F$, the claim is shown.  Otherwise, let $x^> \in F$ be such that $\dF(\tau_{x^>}(P), Q) > \eps$.
    
    Let $R$ be a curve on $F$ such that $R(0) = x^=$ and $R(1) = x^>$, and let $\eps(r) = \dF(\tau_r(P), Q)$ for $r \in R$.  Assume without loss of generality that $\eps(R(s)) > \eps$ for any $s > 0$.  From the definition of discrete \Frechet distance and the continuity of critical transformations as a function of $\eps$, we have that there must exist some open neighborhood $S$ around 0 for which $R(s) \in T^{\eps(R(s))}_\crit(c^*)$ for all $s \in S$ and some configuration $c^*$.  Therefore, since $T^{a}_\crit(c^*) \cap T^{b}_\crit(c^*) = \emptyset$ for $a \ne b$, we have that $R(0) \in T^{\eps}_\crit(c^*)$, but $R(s) \notin T^{\eps}_\crit(c^*)$ for any small value of $s$.  However, $R \subseteq F \subseteq T^{\eps}_\crit(c)$.  Therefore, $c \ne c^*$ and  $x^= \in T^{\eps}_\crit(c) \cap T^{\eps}_\crit(c^*)$.  We then apply the same argument to the lower dimensional face $F \subseteq T^{\eps}_\crit(c) \cap T^{\eps}_\crit(c^*)$ of $A^\eps_\crit$, and the claim follows by induction.

\end{proof}

With the equivalent lemmas in hand, the solution approach remains the same.  However, the time complexity is of course much faster than the continuous version, owing to the smaller number of configurations and the simpler algorithm for the discrete CPSM.  The number of configurations is $O(nk)$, and each sample point in the transformation space can be tested in $O(nk)$ time, leading to a final running time of $O((nk)^{d+1})$ for the decision version and $O((nk)^{d+1}\log(nk))$ for the optimization version.  The parametric search analysis in \cite{Wenk02} for computing the $\eps$ values to check applies straightforwardly.

\section{Discrete CPSM for Translations in $\Real^2$}
The exact algorithms presented in the previous section all rely on a combination of previous algorithms in the literature. However, for the special case of translations in $\Real^2$ for the Discrete TCPSM, we can use a different approach and solve the problem directly, without resorting to either Wenk's algorithm or Maheshwari's.  We refer to the version of the TCPSM where the set of allowed transformations is restricted to only translations as the \textbf{tCPSM}.  While the algorithm presented in Section \ref{sec:tcpsm} can solve this version of the problem in $O(n^3k^3)$ time, the algorithm presented below can solve it in $(n^2k^2\log(nk))$ time.

The Discrete Subset tCPSM can be reformulated as follows: Find a translation of $P$ such that every vertex of $P$ is within $\eps$ of some point in $S$.  Consider the set of all translations $t \in \Real^2$ that put a vertex $P_i$ within $\eps$ of a point $s \in S$.  This set is a disk in the plane.  With all combinations of vertices of $P$ and points in $S$, we have $O(nk)$ such disks.  If we consider each disk to be ``colored'' with a unique color corresponding to its associated vertex of $P$, the problem can be reformulated as follows: Find a single point that lies within at least one disk of every color.  This problem can be solved by a plane-sweep algorithm, which we now describe.

The event points of the plane-sweep algorithm are the tops and bottoms of each of the $O(nk)$ disks, as well as their $O(n^2k^2)$ intersection points.  The intersection of the sweep line with each disk divides the sweep line into intervals.  For each interval, we store a membership array that keeps track of how many disks of each color that interval is inside, as well as a counter that records how many elements of the membership array are non-zero.  When the sweep line reaches the top of a new disk, the interval in which the disk top resides is split into three pieces, each with the copy of the membership array.  The array element of the color of the new disk is incremented in the middle piece's array, which corresponds to the sweep line interval inside the new disk.  If that element was zero, the non-zero counter is incremented.  When the sweep line reaches the bottom of a disk, the corresponding interval is removed and the two adjacent intervals are merged; their membership arrays will be identical.  At an intersection point, the membership arrays and non-zero counters of the three intervals involved are updated accordingly.  If at any time some interval's non-zero counter becomes equal to $n$, the algorithm stops and returns a point of that interval as the solution.  Otherwise, if the sweep line passes the last disk, the algorithm reports that there is no solution.  The full pseudo-code listing is shown Algorithm \ref{algo:translation}.

    \begin{algorithm}[h]
        \caption{Discrete tCPSM in $\Real^2$ $(P, S, \eps)$}
        \label{algo:translation}
        \begin{algorithmic}[1]
            \State Compute the disks corresponding to $P$ and $S$ and their intersection points.
            \State Create a priority queue using the intersection points and the top and bottom of each disk.
            \State Using a balanced binary search tree, initialize a sweep line with a single interval $I$, whose membership array $M_I$ and non-zero counter $Z_I$ are zeroed.
            \While{the queue is non-empty}
                \State Dequeue an event point $p$.
				\State Find the sweep line interval $I$ in which $p$ resides.
                \If{$p$ is the top of a disk with color $c$}
                    \State Split $I$ into $I, J, K$, copying its membership array.
                    \State Increment $M_{J}[c]$.
                    \State Increment $Z_{J}$ if $M_{J}[c]$ was zero.
                    \If{$Z_{J} = n$}
                        \State \Return $p$ as the solution.
                    \EndIf
                \ElsIf{$p$ is the bottom of a disk}
					\State Delete $I$ from the sweepline.
                    \State Merge its two neighbors, deleting the membership array of one.
                \ElsIf{$p$ is the intersection point of two disks of colors $c_1$ and $c_2$}
                    \State Increment or decrement $M_{I}[c_1]$ and $M_{I}[c_2]$ accordingly.
                    \State Increment or decrement $Z_{I}$ accordingly.
                    \If{$Z_{I} = n$}
                        \State \Return $p$ as the solution.
                    \EndIf
                \EndIf
            \EndWhile
            \State \Return that there is no solution.
        \end{algorithmic}
    \end{algorithm}

\subsection{Analysis}
Lines 1 through 3 take $O(n^2k^2)$ time.  The while loop of line 4 iterates over $O(n^2k^2)$ points, taking $O(\log(nk))$ time for each dequeue operation.  Using a balanced binary search tree as the sweep line data structure allows lines 6 and 15 to be accomplished in $O(\log(nk))$ time.  Line 8 takes $O(n)$ time, but is only executed for disk top points, of which there are only $O(nk)$.  The rest of the lines in the while loop are all $O(1)$.  Thus, the total running time of the algorithm is $O(n^2k^2\log(nk))$.  Surprisingly, this is faster than the best known algorithm to solve the problem when both curves are known, which is $O(n^3m^3)$ \cite{Jiang08}.

\subsection{All-points Version}
If we add the additional constraint that every point in $S$ be within $\eps$ from some vertex in $p$, the problem becomes equivalent to the All-Points variation.  To solve this variation, we define each disk to be of two colors, one for its corresponding vertex in $P$ and one for its corresponding point in $S$.  The algorithm is easily modified to increment or decrement two entries in the membership array instead of one when a new disk or a disk intersection point is encountered.  The non-zero counter of an interval will have to rise to $n + k$ to report a valid solution, and copying the membership array in line 8 will take $O(n + k)$ time.  However, the time complexity is still identical, as the cost of iterating through all event points in sorted order still dominates.

\section{Approximation Algorithms} \label{sec:approx}
\subsection{3-Approximation for Continuous All-Points TCPSM} \label{sec:3approx}
In the appendix of \cite{Accisano14}, we presented a 3-factor approximation algorithm for the Continuous All-points version of the CPSM, which is NP-complete.  The approximation algorithm works by first deciding if there is curve that, in addition to having \Frechet distance at most $\eps$ from $P$, visits each point in $S$ at its closest segment.  We call such a curve \emph{NS-compliant}.  

\begin{definition}
Let $\eps = \df(P, Q)$, where $P$ and $Q$ are curves, and for a given point $q$, let $c(q)$ be the index of the segment of $P$ nearest $q$.  $Q$ is said to be \emph{NS-compliant} if, for every point $q$ in the vertex set $V(Q)$, there exists a vertex $Q_i = q$ that is visited at $\overline{P_{c(q)}}$.
\end{definition}

\begin{lemma} \label{lem:approx}
Let $\eps = \df(P, Q)$, where $P$ and $Q$ are curves.  There exists an NS-compliant curve $Q'$ with the same vertex set as $Q$ such that $\df(P,Q') \le 3\eps$.
\end{lemma}
\begin{proof}
Since $Q$ might not be NS-compliant, there may be some vertices of $Q$ that are not visited at their closest segment of $P$; let $s$ be such a point.  Since $s$ is visited at a segment other than its closest, it can be no further than $\eps$ away from $\overline{P_c(s)}$.  Let $p$ be a point in $\overline{P_c(s)}$ that is within $\eps$ of $s$.  Recall that there is always at least one feasible pair for any point on $P$; let $(s', p)$ be feasible.  Then, add $s'$ as a new vertex of $Q$.  Note that the distance between $s$ and $s'$ is at most $2\eps$.  Repeating this process for every point not visited at its closest segment yields a new curve $Q^*$.  Since each new vertex has been added along an existing segment, $\df(P, Q) = \df(P, Q^*)$.

Now, merge each $s'$ with its corresponding $s$ by translating the former to the position of the latter, yielding a new curve $Q'$ with a potentially different \Frechet distance from $P$.  Let $\sigma$ and $\tau$ be reparameterizations of $P$ and $Q^*$, and consider the point $Q^*(\tau(t))$ for some $t \in [0, 1]$, which lies on some segment of $Q^*$.  The endpoints of the corresponding segment in $Q'$ may have been displaced up to $2\eps$, and thus the point $Q'(\tau(t))$ may be up to $2\eps$ away from $Q^*(\tau(t))$.  Therefore, $\norm{P(\sigma(t)), Q'(\tau(t))}$ can be at most $2\eps$ larger than $\norm{P(\sigma(t)), Q^*(\tau(t))}$.  Finally, since the \Frechet distance is the infimum of the maximum distance over all reparameterizations, we have that $\df(P, Q') \le \df(P, Q^*) + 2\eps = 3\eps$.
\end{proof}

An algorithm presented in the appendix of \cite{Accisano14} can be used for deciding whether there exists a curve which visits all points in $S$, has \Frechet distance at most $\eps$ from $P$, and is NS-compliant.  It runs in $O(nk^2)$ time, and the optimal NS-compliant curve can be found in $O(nk^2\log(nk))$ time by way of parametric search.  Since the optimal curve can be made NS-compliant while only increasing its \Frechet distance by a factor of 3, this yields a 3-approximation algorithm.

\begin{theorem}
	The Continuous All-points CPSM can be 3-approximated in $O(nk^2)$ time.
\end{theorem}

When integrating this result into Wenk's framework, it is tempting to simply apply the decision algorithm to check for an NS-compliant curve at each semi-algebraic sample point in the transformation space, and then use parametric search to find the optimal NS-compliant curve.  However, this will not work.  As $P$ is translated or transformed, the closest segment to each point in $S$ can change, which means curves that were NS-compliant for one transformation may be non-compliant for others.  Therefore, even if $\tau(P)$ is exactly the optimal \Frechet distance away from the optimal NS-compliant curve, the curve may not be NS-compliant for that particular value of $\tau$, causing the decision algorithm to return {\sc false}.

The solution to this problem is to remember that our goal is not to find the optimal NS-compliant curve, but to find the optimal unrestricted curve.  To that end, we modify the algorithm so that, at each step $\eps$ of the parametric search, it checks each sample point in the transformation space for an NS-compliant curve with \Frechet distance at most $3\eps$.

Let $Q_\opt$ be the optimal unrestricted curve that visits every point in $S$, and let $\eps_\opt$ be the \Frechet distance from $Q_\opt$ to the optimal transformation of $P$.  As before, the configurations corresponding to $Q_\opt$ will be among those added to the arrangement, and the others will have no effect on the correctness.  For a given parametric search step $\eps \ge \eps_\opt$, there will be at least one sample point $\tau$ in the parameter space for which $\df(\tau(P), Q_\opt) = \eps$.  For this value of $\tau$, Lemma \ref{lem:approx} guarantees that there will be an NS-compliant curve with \Frechet distance at most $3\eps$ from $\tau(P)$, regardless of which point happens to be closest to which segment.  Thus, the parametric search will continue downward and is guaranteed to terminate at some step $\eps$ for which $\eps_\opt/3 \le \eps \le \eps_\opt$, yielding a transformation $\tau$ and a curve $Q$ with $\eps_\opt \le \df(\tau(P), Q) \le 3\eps_\opt$.

\begin{theorem}
The Continuous All-Points TCPSM can be 3-approximated in $O((nk)^{2d+1}k\log(nk))$ time.
\end{theorem}

\subsection{(1+$\eps$)-Approximation}
The running times of the exact algorithms discussed in the previous sections are all quite high, even when the degrees of freedom are few.  Because of this, it makes sense to look at approximation algorithms that might have lower time complexity.  In \cite{Alt01}, the authors present a $(1+\eps)$ approximation algorithm for the two-curve matching problem under translations in $\Real^2$.  In this section, we generalize their approach to work for the tCPSM problem in $\Real^d$.  We make use of the following Lemma, which was proven in \cite{Wenk02}, and a key observation about both types of \Frechet distance.

\begin{lemma} \emph{\hspace*{-.07in}\cite{Wenk02}} \label{lem:translate}
Let $P$ and $Q$ be curves and let $t$ be a translation vector in $\Real^d$.  Then $\df(P + t, Q) \le \df(P, Q) + \norm{t}$.  The same applies to discrete \Frechet distance.
\end{lemma}

\begin{observation}
    Let $P$ and $Q$ be curves.  Then $\norm{P_0, Q_0} \le \df(P, Q) \le \dF(P, Q)$.
\end{observation}

This suggests the strategy of translating $P$ such that its start point overlaps with the start point of the optimal curve $Q$.  In fact, this strategy guarantees a \Frechet distance from $Q$ at most twice the optimal.  To show this, let $t_\opt$ be the optimal translation of $P$ and let $t_\apx = Q_0 - P_0$ be the translation that, when applied to $P$, lines up the two initial points of the curve.  By the observation above, we have that $\norm{t_\apx - t_\opt} = \norm{P_0 + t_\opt, P_0 + t_\apx}  = \norm{P_0 + t_\opt, Q_0} \le \df(P + t_\opt, Q)$.  By Lemma \ref{lem:translate}, this shows that $\df(P + t_\apx, Q)  = \df(P + t_\mathrm{opt} + (t_\apx - t_\opt), Q) \le 2\df(P + t_\opt, Q)$.  Finally, we can improve this 2-approximation to a $1+\eps$ approximation by centering an $\eps$-lattice of width 2 around $Q_0$ and trying every lattice point.

Of course, all this assumes that the optimal curve $Q$ is known, which it is not.  However, its start point must be one of the $k$ input points.  Trying each point in $S$ adds another factor of $k$, leading to a final worst case running time of $O(nk^3\log(nk)/\eps^{d})$ for the continuous versions and $O(nk^2/\eps^{d})$ for the discrete.  We note that, in practice, this additional factor of $k$ will typically have a very small associated constant, since those translations that do not also put the end point of $P$ close to a point in $S$ can quickly be ruled out.

\begin{theorem}
	The tCPSM can be $(1+\eps)$-approximated in $O(nk^3\log(nk)/\eps^{d})$ time for the Continuous Subset version and $O(nk^2/\eps^{d})$ for the Discrete Subset and Discrete All-points versions.
\end{theorem}

In addition, the same idea can be applied to the 3-approximation algorithm for the Continuous All-points CPSM discussed in Section \ref{sec:3approx} to yield a $(3+\eps)$ approximation for the corresponding tCPSM version.

\begin{corollary}
	The Continuous All-points tCPSM can be $(3+\eps)$-approximated in  $O(nk^3\log(nk)/\eps^{d})$.
\end{corollary}

\section{Conclusion}
In summary, we have shown how various versions of the TCPSM can be solved in polynomial time.  The Continuous Subset version can be solved exactly and the Continuous All-Points version can be 3-approximated, both in $O((nk)^{2d+1}k\log(nk)$ time.  Both Discrete versions can be solved in $O((nk)^{d+1}\log(nk)$ time, and for the special case of translations in $\Real^2$, $O(n^2k^2\log^2(nk))$ time.  Furthermore, a $(1+\eps)$ approximation algorithm exists for translations in $\Real^d$ that runs in $O(nk^3\log(nk))$ time.

\bibliographystyle{plain}
\bibliography{frechet}

\begin{thebibliography}{10}

\bibitem{Accisano12}
P.~Accisano and A.~{\"U}ng{\"o}r.
\newblock Hardness results on curve/point set matching with {F}r{\'e}chet
  distance.
\newblock In {\em Proc.\ of the 29th European Workshop on Comp.\ Geom.}, 2013.

\bibitem{Accisano14}
P.~Accisano and A.~{\"U}ng{\"o}r.
\newblock Matching curves to imprecise point sets using {F}r\'echet distance.
\newblock {\em arXiv preprint arxiv.org}, 2014.

\bibitem{Agarwal12}
Pankaj~K Agarwal, Rinat~Ben Avraham, Haim Kaplan, and Micha Sharir.
\newblock Computing the discrete fr\'echet distance in subquadratic time.
\newblock {\em arXiv preprint arXiv:1204.5333}, 2012.

\bibitem{Alt03}
H.~Alt, A.~Efrat, G.~Rote, and C.~Wenk.
\newblock Matching planar maps.
\newblock In {\em Proc. of the 14th annual ACM-SIAM Symp. on Discrete
  Algorithms}, pages 589--598, 2003.

\bibitem{Alt95}
H.~Alt and M.~Godau.
\newblock Computing the {F}r{\'e}chet distance between two polygonal curves.
\newblock {\em Int.\ J.\ of Comp.\ Geom.\ \&\ Appl.}, 5(01n02):75--91, 1995.

\bibitem{Alt01}
H.~Alt, C.~Knauer, and C.~Wenk.
\newblock Matching polygonal curves with respect to the {F}r{\'e}chet distance.
\newblock In {\em STACS 2001}, pages 63--74. Springer, 2001.

\bibitem{Basu98}
S.~Basu, R.~Pollack, and M.~Roy.
\newblock A new algorithm to find a point in every cell defined by a family of
  polynomials.
\newblock In {\em Quantifier elimination and cylindrical algebraic
  decomposition}, pages 341--350. Springer, 1998.

\bibitem{Brakatsoulas05}
S.~Brakatsoulas, D.~Pfoser, R.~Salas, and C.~Wenk.
\newblock On map-matching vehicle tracking data.
\newblock In {\em Proc.\ Int.\ Conf.\ on Very Large Data Bases}, pages
  853--864, 2005.

\bibitem{Eiter94}
T.~Eiter and H.~Mannila.
\newblock Computing discrete {F}r{\'e}chet distance.
\newblock Technical Report CD-TR 94/65, Christian Doppler Laboratory, 1994.

\bibitem{Ewing85}
G.~Ewing.
\newblock {\em Calculus of variations with applications}.
\newblock Dover Pub., New York, 1985.

\bibitem{Jiang08}
M.~Jiang, X.~Ying, and B.~Zhu.
\newblock Protein structure--structure alignment with discrete {F}r\'echet
  distance.
\newblock {\em J. of Bioinfo. and Comp. Biology}, 6(01):51--64, 2008.

\bibitem{Maheshwari11}
A.~Maheshwari, J.~Sack, K.~Shahbaz, and H.~Zarrabi-Zadeh.
\newblock Staying close to a curve.
\newblock In {\em Canadian Conf.\ on Comp.\ Geom.}, pages 55--58, 2011.

\bibitem{Wenk02}
C.~Wenk.
\newblock {\em Shape matching in higher dimensions}.
\newblock PhD thesis, Freie Universit{\"a}t Berlin, Universit{\"a}tsbibliothek,
  2002.

\bibitem{Wylie12}
T.~Wylie.
\newblock Discretely following a curve.
\newblock In {\em Combinatorial Optimization and Applications}, pages 13--24.
  Springer, 2013.

\end{thebibliography}

\end{document}